\documentclass[9pt, twocolumn, journal]{IEEEtran}
\input{preamble_ieee.sty}
\usepackage{mathtools}
\usepackage{flushend}
\title{Computing Entropy Rate Of Symbol Sources \\ \& A Distribution-free Limit Theorem} 
\author{
\begin{tabular}{cc}
Ishanu Chattopadhyay & Hod Lipson \\ \texttt{ic99@cornell.edu} & \texttt{hod.lipson@cornell.edu}
\end{tabular}
}
\begin{document}  
\maketitle 
\begin{abstract}
Entropy rate of sequential data-streams naturally quantifies the  complexity of the generative process. Thus entropy rate fluctuations   could be used as a  tool to recognize  dynamical perturbations in  signal sources, and could  potentially be carried out without explicit   background noise characterization. However, state of the art algorithms to estimate the entropy rate have markedly slow convergence; making such entropic approaches non-viable in practice. We present here a fundamentally new approach to estimate entropy rates, which is demonstrated to converge significantly faster in terms of input data lengths, and is shown to be effective in diverse applications ranging from the estimation of the  entropy rate of English texts to the estimation of  complexity of chaotic dynamical systems. Additionally, the convergence rate of  entropy estimates  do not follow from any standard limit theorem, and reported algorithms  fail to provide  any  confidence bounds on the computed values.  Exploiting a connection to the theory of probabilistic automata, we establish a  convergence rate of $O(\log \vert s \vert/\sqrt[3]{\vert s \vert})$ as a function of the input length $\vert s \vert$,  which then   yields explicit  uncertainty estimates, as well as  required  data lengths to satisfy pre-specified confidence bounds.
\end{abstract}
\allowdisplaybreaks{ 
%

\section{Motivation, Background \& Contribution}
The entropy rate of a stationary and ergodic process converges in probability to the  per-letter Kolmogorov complexity of a single sufficiently long sample path~\cite{horibe03}. While  Kolmogorov complexity is incomputable, entropy rates  can, in principle, be estimated. Ability to quantify  the complexity of a signal source, even in the average sense,  can provide valuable insights into the driving dynamics; and can potentially be used as a  tool to detect dynamical anomalies without explicit knowledge of  background noise processes.

However,  source entropy rate estimation  from an observed sample path is computationally non-trivial. Even with the assumptions of ergodicity and stationarity, one cannot fruitfully apply the defining relation in Eq.\eqref{eqentropy} due to the exponential increase in the number of different words with the word-length. This is particularly important  if there are long-range dependencies in the symbol stream.  Such dependencies introduce additional long-range structure; decreasing the source entropy  in the process. In such cases  unacceptably long words or \textit{blocks} must be considered, and pre-mature truncation of the computation would lead to large errors. 

 The best known algorithms  that carry out a  more efficient  computation are based on  Lempel-Ziv (LZ) source coding~\cite{LZ77,LZ78,LZ77opt}. The LZ coding algorithms are asymptotically optimal, $i.e.$ their compression rate approaches the source entropy rate for any ergodic stationary stochastic process. The key idea here is adaptive dictionary compression: parse the input string into distinct phrases, and represent them with codewords, making sure that  short codewords are assigned to common phrases. Done optimally, one ends up with a compressed string, such that the ratio of the input and output lengths approach the source entropy rate. Different variations on this idea have been reported~\cite{langdon83,grassberger89}. Techniques distinct from  LZ parsing  are also known, $e.g.$,  Rissanen~\cite{rissanen83} reported a universal compression scheme, which instead of gathering parsed segments of the input along with their occurrence counts, collects  the ``contexts'' in which each symbol of  the  input string  occurs, together with   conditional  occurrence counts.

Importantly, a majority of the reported techniques do more than just compute the entropy rate; they are indeed full-scale data compression utilities, that produce a decodable representation of the input. Can we do better if we are only interested in the former? This paper provides an affirmative answer to this possibility.

Secondly,  existing techniques lack convergence rate estimates; computation of error bars for reported approaches do not follow from any standard limit theorem. There is indeed no analytical way to check for the internal consistency of
the estimation or its accuracy. We may observe gradual convergence to a limiting value, and this is indeed guaranteed by theory;  but  are unable to  provide uncertainty bounds on the computed  estimate with finite inputs. Typically observed slow convergence in all non-trivial scenarios, for all reported algorithms,  makes this a key issue. An empirical relationship, without proof or theoretical backing, has been suggested~\cite{shgs96}, which conjectures the $\vert s \vert$-dependence ($\vert s \vert$ being the length of the input $s$) of the estimated entropy rate $\widetilde{H}$ to follow
$\widetilde{H} \simeq H_{actual} + c \frac{\log \vert s \vert}{\vert s \vert^\gamma}, \textrm{where } c,\gamma \textrm{ are fit parameters}$. 
In this paper, we show that, at least with our algorithm, the convergence rate is given by $O(\log \vert s \vert / \sqrt[3]{\vert s \vert})$. This is a distribution-free result, in the sense that the asymptotic  bound does not depend on the source characteristics. In consequence, we can derive explicit uncertainty estimates at specified  confidence bounds on the estimated entropy rate for finite-length input data.

\begin{figure}[t]
\tikzexternalenable
\centering
 \begin{tikzpicture}[,font=\bf\fontsize{8}{8}\selectfont,->,>=latex',shorten >=1pt,auto,node distance=1.2cm,  scale=1 ,]
  \node [circle, draw=black, dashed, thick ] (A) at (0.15,1.5) [align=center]{\bf Hidden\\\bf Process};
\draw [fill=lightgray,draw=none,opacity=.9, xslant=0.0,yslant=.4, path fading=east] (0,0) -- (1.5,0) -- (1.5,2.5) -- (0,2.5) -- cycle;  
\draw [xshift=-.8in,yshift=.275in,fill=lightgray,draw=none,opacity=.9, xslant=0.0,yslant=-.35, path fading=west] (0,0) -- (2,0) -- (2,2.5) -- (0,2.5) -- cycle;  
\draw [xshift=-.635in,yshift=1.225in,fill=lightgray,draw=none,opacity=1, xslant=2.45,yslant=-.19, path fading=west] (0,0) -- (3,0) -- (3,.5) -- (0,.5) -- cycle; 


    \node [anchor=west] (B) at (A.east) [align=center]{\bf 01010010101011100101010101010$\boldsymbol{\cdots}$};
   \node [anchor=north, text=Red4] (C) at ([yshift=-.1in]B.south) [align=left]{\textbullet~Entropy rate?};
   \node [anchor=north west, text=Red4] (C1) at ([yshift=0.07in]C.south west) [align=left]{\textbullet~Convergence Rate of estimate?};  
 \node [anchor=south, text=gray] (C) at (B.north) [align=center]{States are\\ hidden};
   \node [anchor=north east,text=gray,font=\bf\fontsize{7}{8}\selectfont] (D) at ([yshift=.45in,xshift=-.1in]A.south west) [align=center]{ Stationary\\ergodic\\stochastic\\process};
\end{tikzpicture} 
\vspace{-5pt}

\captionN{\textbf{Problem description.} Given a quantized data stream, how do we compute the entropy rate of the hidden  process? Even with the  assumption of  stationarity and ergodicity for  the generator, reported algorithms converge very slowly. Additionally,  these convergence rates are unknown for such approaches; implying that we cannot put uncertainty bounds on the computed values in practice. We show that a significantly faster computation of the entropy rate is possible; and derive a universal lower bound on how slowly this convergence might occur.
}\label{figstrm}
\vspace{-12pt}

\end{figure}
\subsection{Key Insight}
Our approach is based on  modeling discrete and finite-valued stationary and ergodic sources as probabilistic automata. Our automata is distinct from that of Paz~\cite{P71}, and each model in our case is  in fact an  encoding of a measure defined on the space of  strictly infinite strings over a finite alphabet. While the formalisms are  completely different, some aspects of this approach has subtle parallels to that of Rissanen's ``context algorithm''~\cite{rissanen83}; his search for contexts which yield similar probabilities of generating future symbols is analogous to our search for a synchronizing string in the input stream - a finite sequence of symbols that, once executed on a probabilistic automaton, leads to a fixed state irrespective of the initial conditions. Of course we do not know anything about the hidden model a priori; but nevertheless we establish that such a string, at least in a well-defined approximate sense, always exists and is identifiable efficiently. Finally, we show that, given such an approximate synchronizing string, we can use results from non-parametric statistics to bound the probability of error as a function of the input length.

\subsection{Entropy  \& Entropy Rate}
Entropy $H(X)$ of a discrete random variable $X$, taking values in the alphabet $\Sigma$,  is defined as:
\cgather{
H(X) = - \sum_{x \in \Sigma} p(x)\log p(x)
}where $p(x)$ is the probability of occurrence of $x \in \Sigma$. 
The base of the logarithm is generally taken to be $2$, and then the entropy is being expressed in \textit{bits}. While the definition of entropy of a random variable may be obtained axiomatically, a perhaps more compelling  approach is to show that it arises 
as the average length of the shortest description of a random variable~\cite{cover}.

The joint entropy of a set of random variables $X_1,\cdots,X_n$, with $X_i$  taking values in the alphabet $\Sigma_i$, is defined in the usual manner: 
\cgather{
H(X_1,\cdots,X_n) = - \sum_{x_i \in \Sigma_i} p(x_1,\cdots, x_n) \log_2 p(x_1,\cdots, x_n) 
}
The chain rule for entropy calculations~\cite{cover} follows from the definitions, and is of particular importance:
\cgather{\label{eqchrule}
H(X_1,\cdots,X_n)= \sum_{i=1}^n H(X_i \vert X_{i-1},\cdots,X_1)
}
The notion of entropy formalizes the Asymptotic Equipartion Property (AEP): If discrete random variables  $X_1,\cdots, X_n$ are i.i.d. and have probability mass function $p(x)$, then we have:
\cgather{
-\frac{1}{n} \log p(X_1,\cdots, X_n) \xrightarrow{a.s} H(X) = - \sum_{x \in \Sigma} p(x)\log p(x)
}
The AEP implies that $nH(X)$ bits suffice on average to describe $n$ i.i.d. random variables. If the random variables are not independent,
the entropy $H(X_1,\cdots,X_n)$ still grows asymptotically linearly with $n$ at a rate known as the entropy rate of the process. In particular, if the random variables  define a stationary ergodic stochastic process $\mathcal{X} = \{X_i\}$, then the AEP still holds:
\cgather{
-\frac{1}{n}\log p(X_1,\cdots, X_n) \xrightarrow{\textrm{a.s.}} H(\mathcal{X})
}
where $H(\mathcal{X})$ is the entropy rate of the process defined as:
%
\cgather{\label{eqentropy}
H(\mathcal{X}) = \lim_{n \rightarrow \infty} \frac{1}{n} H(X_1,\cdots,X_n)
}
As in the case of the i.i.d variables,  \textit{typical sequences} of length $n$  may be represented using approximately $nH(\mathcal{X})$ bits. Thus the entropy rate  quantifies the average description length of the process, and hence its expected complexity~\cite{horibe03}. 

\section{Stochastic Processes \& Probabilistic Automata}\label{sec3}
As mentioned earlier, our approach hinges upon effectively using probabilistic automata to model stationary, ergodic processes. Our automata models are distinct to those reported in the literature~\cite{P71,VTCC05}.  The details of this  formalism can be found in \cite{CL12g}; we include a brief overview here for the sake of completeness.

\begin{notn} 
 $\Sigma$  denotes a finite alphabet  of symbols. The set of all finite but possibly unbounded strings on $\Sigma$ is denoted by $\Sigma^\star$~\cite{HMU01}. The set of finite strings over $\Sigma$ form a concatenative monoid, with the empty word $\lambda$ as identity. 
The set of strictly infinite strings on $\Sigma$ is denoted as $\Sigma^\omega$, where $\omega$ denotes the first transfinite cardinal. 
For a string $x$,  $\vert x \vert$ denotes its length, and for a set $A$,    $\vert A \vert$ denotes its cardinality. Also, $\Sigma^d_+ = \{x \in \Sigma^\star \textrm{ s.t. } \vert x \vert  \leqq d\}$.
\end{notn}

\begin{defn}[QSP]\label{defQSP}
A QSP $\mathcal{H}$ is a discrete time $\Sigma$-valued strictly stationary, ergodic stochastic process, $i.e.$ 
\cgather{
\mathcal{H} = \left \{ X_t: X_t \textrm{ is a $\Sigma$-valued random variable}, t \in \mathbb{N}\cup \{0\} \right \}
} 
A  process is ergodic if  moments may be calculated from a sufficiently long realization, and strictly stationary if moments are time-invariant.
\end{defn}
We next formalize the connection of QSPs to PFSA generators.
  We develop the theory assuming multiple realizations of the QSP $\mathcal{H}$, and  fixed initial conditions. Using ergodicity, we will be then able to apply our construction to a single sufficiently long realization, where  initial conditions cease to matter.
\begin{defn}[$\sigma$-Algebra On Infinite Strings]
 For the set of infinite strings on  $\Sigma$, we define $\mathfrak{B}$ to be the smallest $\sigma$-algebra generated by the family of sets $\{  x \Sigma^\omega : x \in \Sigma^\star\}$.
\end{defn}
\begin{lem}\label{QSPtoProb}
Every QSP  induces a  probability space $(\Sigma^\omega,\mathfrak{B},\mu)$.
\end{lem} 
\begin{proof}
Assuming stationarity, we can construct a probability measure $\mu: \mathfrak{B} \rightarrow [0,1]$  by defining
for any sequence $x\in \Sigma^\star\setminus \{\lambda\}$, and a sufficiently large number of realizations $N_R$ (assuming ergodicity):
\cgathers{
 \mu( x \Sigma^\omega) = \lim_{N_R \rightarrow \infty}\frac{ \textrm{\small \#  of initial occurrences of $x$}}{\begin{array}{c} \textrm{\small \#  of initial occurrences} \\\textrm{\small of all sequences of length  $\vert x \vert$}\end{array}}
}
and extending the measure to elements of $\mathfrak{B} \setminus B$ via at most countable sums. Thus $\mu(\Sigma^\omega) = \sum_{x \in \Sigma^\star} \mu( x \Sigma^\omega) = 1$, and for the null word  $\mu(\lambda \Sigma^\omega) =  \mu(\Sigma^\omega) = 1$.
\end{proof}

 \begin{notn}
 For notational brevity, we denote $\mu( x \Sigma^\omega)$ as $Pr(x)$.
 \end{notn}

Classically,  automaton states are equivalence classes for the  Nerode relation;  two strings are  equivalent if and only if any finite extension of the strings is either both in the language under consideration, or neither are~\cite{HMU01}. We use a probabilistic extension~\cite{CR08}.

\begin{defn}[Probabilistic Nerode Equivalence Relation]\label{defnerode} $(\Sigma^\omega,\mathfrak{B},\mu)$ induces an equivalence relation $\sim_{N}$ on the set of finite strings $\Sigma^\star$ as:
\cgather{
\forall x,y \in \Sigma^\star, 
\smash{x \sim_{N} y \iff \forall z \in \Sigma^\star \bigg (\big (} Pr(xz)=Pr(yz)=0 \big )  \notag \\  \bigvee \big \vert  Pr(xz)/Pr(x)- Pr(yz)/Pr(y)\big \vert =0\bigg )
}
\end{defn}

\begin{notn}
For $x \in \Sigma^\star$,  the equivalence class of $x$ is  $[x]$.
\end{notn}

It is easy to see   that $\sim_{N}$ is right invariant, $i.e.$ 
\cgather{
x \sim_{N} y \Rightarrow \forall z \in \Sigma^\star, xz \sim_{N} yz 
}
A right-invariant equivalence on $\Sigma^\star$ always induces an automaton structure; and hence the probabilistic Nerode relation induces  a probabilistic automaton: states are equivalence classes of $\sim_{N}$, and the transition structure arises as follows: For states $q_i,q_j$, and  $x \in \Sigma^\star$,
\begin{gather}
([x]=q ) \wedge ([x \sigma ] = q' 
 )\Rightarrow q \xrightarrow{\sigma} q'
\end{gather}
Before formalizing the above construction, we introduce the notion of probabilistic automata with initial, but no final, states.


\begin{defn}[Initial-Marked PFSA]\label{defpfsa} An initial marked probabilistic finite state automaton (a Initial-Marked PFSA)   is a quintuple $(Q,\Sigma,\delta,\pitilde,q_0)$, where $Q$ is a finite  state set, $\Sigma$ is the alphabet,  $\delta:Q \times \Sigma \rightarrow Q$ is the  state transition function,  $\pitilde : Q \times \Sigma \rightarrow [0,1]$  specifies the conditional symbol-generation probabilities, and $q_0\in Q$ is the initial state.  $\delta$ and $\pitilde$ are recursively extended to arbitrary  $y=\sigma x \in \Sigma^\star$ as follows:
\cgather{
\forall q \in Q, \delta(q,\lambda) = q\\
\delta(q,\sigma x) = \delta(\delta(q,\sigma),x)\\
\forall q \in Q, \pitilde(q,\lambda) = 1\\
\pitilde(q,\sigma x) = \pitilde(q,\sigma)\pitilde(\delta(q,\sigma),x)
}
Additionally, we impose  that for  distinct states $q_i,q_j \in Q$, there exists a string $x \in \Sigma^\star$, such that $\delta(q_i,x) = q_j$, and $\pitilde(q_i,x) > 0$.
\end{defn}

Note that the probability of the null word is unity from each state.

If the current state and  the next symbol is  specified, our next  state is fixed; similar to Probabilistic Deterministic Automata~\cite{Gavalda06}. However, unlike the latter, we lack final states in the model. Additionally, we assume our graphs to be strongly connected.

Later we will remove initial state dependence using ergodicity. 
Next we formalize how  a PFSA arises  from  a  QSP.

\begin{lem}[PFSA Generator]\label{lemPFSAgen}
Every Initial-Marked PFSA $G=(Q,\Sigma,\delta,\pitilde,q_0)$ induces a unique probability measure $\mu_G$ on the measurable space  $(\Sigma^\omega,\mathfrak{B})$.
\end{lem}
\begin{proof}
Define   set function $\mu_G$ on the measurable space  $(\Sigma^\omega,\mathfrak{B})$:
\cgather{
\mu_G(\varnothing) \triangleq 0\\
\forall x \in \Sigma^\star, \mu_G(x\Sigma^\omega) \triangleq \pitilde(q_0,x)\\
\forall x,y \in \Sigma^\star, \mu_G(\{x,y\}\Sigma^\omega) \triangleq \mu_G(x\Sigma^\omega)+ \mu_G(y\Sigma^\omega)
}
Countable additivity of  $\mu_G$ is immediate, and  (See Definition~\ref{defpfsa}):
\cgather{
\mu_G(\Sigma^\omega) = \mu_G(\lambda\Sigma^\omega) = \pitilde(q_0, \lambda) = 1
}
implying that $(\Sigma^\omega,\mathfrak{B}, \mu_G)$ is a probability space.
\end{proof}

We refer to $(\Sigma^\omega,\mathfrak{B}, \mu_G)$ as the probability space generated by the Initial-Marked PFSA $G$. 

\begin{lem}[Probability Space To PFSA]\label{lemPROB2PFSA}
If the probabilistic Nerode relation corresponding to a  probability space $(\Sigma^\omega,\mathfrak{B}, \mu)$ has a finite index, then the latter has an  initial-marked PFSA generator.
\end{lem}
\begin{proof}
Let  $Q$ be the set of equivalence classes of the probabilistic Nerode relation (Definition~\ref{defnerode}),  and define functions $\delta:Q \times \Sigma \rightarrow Q$, $\pitilde:Q \times \Sigma \rightarrow [0,1]$ as:
\calign{
& \delta([x],\sigma) = [x\sigma]\\
& \pitilde([x],\sigma) = \frac{Pr(x'\sigma)}{Pr(x')} \textrm{ for any choice of } x' \in [x] \label{eqpit}
}
where we  extend $\delta,\pitilde$  recursively  to   $y=\sigma x \in \Sigma^\star$ as 
\cgather{\delta(q,\sigma x) = \delta(\delta(q,\sigma),x)\\\pitilde(q,\sigma x) = \pitilde(q,\sigma)\pitilde(\delta(q,\sigma),x)}
For verifying the null-word probability, choose a $x \in \Sigma^\star$ such that $[x] = q$ for some $q \in Q$. Then, from  Eq.~\eqref{eqpit}, we have:
\cgather{
 \displaystyle \pitilde(q,\lambda)= \frac{Pr(x'\lambda)}{Pr(x')} \textrm{ for any  } x' \in [x] \Rightarrow \pitilde(q,\lambda)
=  \frac{Pr(x')}{Pr(x')} = 1
}
Finite index of $\sim_{N}$ implies  $\vert Q\vert < \infty$, and hence denoting  $[\lambda]$ as $q_0$, we conclude:  $G=(Q,\Sigma,\delta,\pitilde,q_0)$ is an Initial-Marked PFSA.  Lemma~\ref{lemPFSAgen} implies that $G$ generates $(\Sigma^\omega,\mathfrak{B}, \mu)$,  which completes the proof.
\end{proof}

The above construction yields a \textit{minimal realization} for the Initial-Marked PFSA, unique up to state renaming. 


\begin{lem}[QSP to PFSA]\label{lemQSP2PFSA}
Any QSP  with a finite index Nerode equivalence is generated by an Initial-Marked PFSA.
\end{lem}
\begin{proof}
Follows immediately from Lemma~\ref{QSPtoProb} (QSP to Probability Space) and Lemma~\ref{lemPROB2PFSA} (Probability Space to PFSA generator).
\end{proof}

\subsection{Canonical Representations}
 We have defined a QSP as both ergodic and stationary, whereas the Initial-Marked PFSAs have a designated initial state. Next we introduce  canonical representations to remove initial-state dependence. We use $\Pitilde$ to denote the matrix representation of  $\pitilde$, $i.e.$, $\Pitilde_{ij} = \pitilde(q_i,\sigma_j)$,  $q_i \in Q, \sigma_j \in \Sigma$. We need the notion of transformation matrices $\Gamma_\sigma$.

\begin{defn}[Transformation Matrices]\label{defGamma}
 For an initial-marked PFSA $G=(Q,\Sigma,\delta,\pitilde,q_0)$, the symbol-specific transformation matrices $\Gamma_\sigma \in \{0,1\}^{\vert Q \vert \times \vert Q \vert}$ are:
\cgather{
\Gamma_\sigma \big \vert_{ij} = \begin{cases}
                                 \pitilde(q_i,\sigma), & \textrm{if } \delta(q_i,\sigma) = q_j \\
				 0, & \textrm{otherwise}
                                \end{cases}
}
\end{defn}
Transformation matrices have a single non-zero entry per row, reflecting our generation rule that given a state and a generated symbol, the next state is fixed. 

First, we note that, given an initial-marked PFSA $G$,  we can associate a probability distribution $\wp_x$ over the states of $G$  for each $x \in \Sigma^\star$ in the following sense:
if $x=\sigma_{r_1}\cdots \sigma_{r_m} \in \Sigma^\star$, then we have:
\cgather{
\wp_x = \wp_{\sigma_{r_1}\cdots \sigma_{r_m}} = \underbrace{\frac{1}{\vert \vert \wp_\lambda \prod_{j=1}^m \Gamma_{\sigma_{r_j}}  \vert \vert_1 }}_{\textrm{Normalizing factor}}\wp_\lambda \prod_{j=1}^m \Gamma_{\sigma_{r_j}}
}
where $\wp_\lambda$ is the stationary distribution over the states of $G$.  Note that there may exist more than one string that leads to a distribution $\wp_x$, beginning from the stationary distribution $\wp_\lambda$. Thus, $\wp_x$   is an equivalence class of strings, $i.e.$, $x$ is  not unique. 
\begin{defn}[Canonical Representation]\label{defcanon}

 An  initial-marked PFSA $G=(Q,\Sigma,\delta,\pitilde,q_0)$ uniquely induces a canonical representation $(Q^C,\Sigma,\delta^C,\pitilde^C)$, where $Q^C$ is  a subset of the  set of probability distributions over  $Q$, and  $\delta^C: Q^C \times \Sigma \rightarrow Q^C$,  $\pitilde^C: Q^C \times \Sigma \rightarrow [0,1]$ are constructed  as follows:
\begin{enumerate}
 \item Construct the stationary distribution on $Q$ using the transition probabilities of the Markov Chain induced by $G$, and include this as the first element $\wp_\lambda$ of $Q^C$. Note that the transition matrix for $G$ is the row-stochastic matrix $M \in [0,1]^{\vert Q \vert \times \vert Q \vert}$, with
$
M_{ij} = \sum_{\sigma: \delta(q_i,\sigma)=q_j}\pitilde(q_i,\sigma)
$, and hence $\wp_\lambda$ satisfies:
\cgather{
\wp_\lambda M = \wp_\lambda
}
\item Define  $\delta^C$ and $\pitilde^C$ recursively:
\calign{
&\delta^C(\wp_x, \sigma) = \frac{1}{\vert \vert \wp_x \Gamma_\sigma \vert \vert_1}\wp_x \Gamma_\sigma \triangleq \wp_{x\sigma}\\
&\pitilde^C(\wp_x,\sigma) = \wp_x \Pitilde
}
\end{enumerate}
\end{defn}
For a QSP $\mathcal{H}$, the canonical representation is denoted as $\mathcal{C}_\mathcal{H}$.
\begin{lem}[Properties of Canonical Representation]\label{lemcanonstruc}
 Given an initial-marked PFSA $G=(Q,\Sigma,\delta,\pitilde,q_0)$:
\begin{enumerate}
\item The canonical representation is independent of the initial state.
\item The canonical representation $(Q^C,\Sigma,\delta^C,\pitilde^C)$ contains a copy of $G$ in the sense that there exists a set of states $Q' \subset Q^C$, such that there exists a one-to-one map $\zeta:Q \rightarrow Q'$, with:
\cgather{\forall q \in Q, \forall \sigma \in \Sigma, \left \{ \begin{array}{l}
\pitilde(q,\sigma) = \pitilde^C(\zeta(q),\sigma) \\
\delta(q,\sigma) = \delta^C(\zeta(q),\sigma)\end{array}\right.
}
\item If during the  construction (beginning with $\wp_\lambda$) we  encounter $\wp_x = \zeta(q)$ for some $x \in \Sigma^\star$,  $q \in Q$ and any map $\zeta$ as defined in (2), then we stay within the graph of the copy of the  initial-marked PFSA  for all right extensions of $x$.
\end{enumerate}

\end{lem}
\begin{proof}
(1) follows  the ergodicity of QSPs, which makes $\wp_\lambda$ independent of the initial state in the initial-marked PFSA.

(2) The canonical representation  subsumes the initial-marked representation in the  sense that the states of the latter may themselves be seen as degenerate distributions over $Q$, $i.e.$, by letting 
\cgather{\label{eqE}
\mathcal{E}=\big \{ e^i \in [0 \ 1]^{\vert Q \vert } , i= 1,\cdots, \vert Q \vert \big \}
}
  denote  the set of distributions satisfying:
\cgather{
e^i\vert_j = \begin{cases}
              1, & \textrm{if } i=j\\
              0, & \textrm{otherwise}
             \end{cases}
}
(3) follows from the strong connectivity of $G$.
\end{proof}



Lemma~\ref{lemcanonstruc} implies that initial states are unimportant;   we may denote the initial-marked PFSA induced by a QSP $\mathcal{H}$, with the initial marking removed, as $\mathcal{P}_\mathcal{H}$, and refer to it simply as a ``PFSA''.  States in $\mathcal{P}_\mathcal{H}$ are representable as  states in $\mathcal{C}_\mathcal{H}$ as elements of $\mathcal{E}$. Next we show that  we always  encounter a state arbitrarily close to some element in  $\mathcal{E}$ (See Eq.~\eqref{eqE}) in the canonical construction starting from the stationary distribution $\wp_\lambda$ on the states of $\mathcal{P}_\mathcal{H}$.
%

Next we introduce the notion of $\epsilon$-synchronization of probabilistic automata (See Figure~\ref{figsync}), which would be of fundamental importance to our entropy estimation algorithm in the next section. Synchronization of automata is fixing or determining the current state; thus it is analogous to contexts in Rissanen's ``context algorithm''~\cite{rissanen83}. We show that not all PFSAs are synchronizable, but all are $\epsilon$-synchronizable.
\begin{figure}[t]
\tikzexternaldisable
\centering 
\newcommand{\Vdist}{.5in}
\def\COL{black!50}
\def\COLA{Red4!20!black}
\def\COLB{Red4!20!black}
\def\COLC{Red4!20!black}
\def\GCOL{gray!0}

\definecolor{blcol}{RGB}{150,150,250}
\definecolor{murcol}{RGB}{255,150,150}
\definecolor{healcol}{RGB}{150,235,150}
\definecolor{colzeta}{RGB}{250,250,200}
\definecolor{leftpcol}{RGB}{200,255,200}
\definecolor{rightpcol}{RGB}{200,200,255}
\definecolor{boxcol}{RGB}{210,200,200}
\definecolor{linecol}{RGB}{200,180,180}
\definecolor{colplus}{RGB}{235,220,220}
\definecolor{colinv}{RGB}{200,220,255}
\definecolor{cof}{RGB}{219,144,71}
\definecolor{pur}{RGB}{200,200,200}
\definecolor{greeo}{RGB}{91,173,69}
\definecolor{greet}{RGB}{52,111,72}
 \definecolor{nodecol}{RGB}{180,180,220}
 \definecolor{nodeedge}{RGB}{140,148,155}
 \definecolor{nodeedgeb}{RGB}{240,148,155}
  \definecolor{nodecolb}{RGB}{220,180,180}
  \definecolor{nodecolc}{RGB}{180,220,180}
  \definecolor{nodecolcD}{RGB}{100,160,100}
  \definecolor{nodecolW}{RGB}{190,190,190}
  \definecolor{edgecol}{RGB}{60,60,80}
 \definecolor{nodecolD}{RGB}{140,140,180}
  \definecolor{nodecolbD}{RGB}{180,140,140}
\tikzset{srule/.style={ semithick, opacity=.8, Red1, text opacity=1, text=gray}}
\tikzset{axstyle/.style={ black, opacity=1,  thick, rounded corners=0pt}}
\tikzset{%
fshadow/.style={      preaction={
         fill=black,opacity=.2,
         path fading=circle with fuzzy edge 20 percent,
         transform canvas={xshift=1mm,yshift=-1mm}
       }}
}
\tikzset{oplus/.style={path picture={%
      \draw[black]
       (path picture bounding box.south) -- (path picture bounding box.north) 
       (path picture bounding box.west) -- (path picture bounding box.east);
      }}} 
\tikzset{%
  highlight/.style={draw=Red4,rectangle,rounded corners=1pt, opacity=1,fill=gray!10,thick,inner sep=-1.25pt,on background layer,fill opacity=.3}
}
\tikzset{%
  highlightg/.style={draw=DodgerBlue4,rectangle,rounded corners=1pt, opacity=1,fill=gray!10,thick,inner sep=-1.25pt,on background layer,fill opacity=.3}
}
%
\def\SCALE{1.5}
\newcommand{\Gone}{%
\begin{tikzpicture}[->,>=stealth',shorten >=1pt,auto,node distance=1.2cm,
                    semithick,scale=\SCALE,font=\bf\fontsize{6}{6}\selectfont]
  \tikzstyle{every state}=[fill=nodecol,draw=nodeedge,text=black,minimum size=3, text width=2,scale=\SCALE,fshadow]
  \node[state] (A)         []           {$\mspace{-6mu}q_0$};
  \node[state]         (B) [right of=A] {$\mspace{-6mu}q_1$};
  \path (A) edge   [draw=edgecol,bend left]           node {$\sigma_1\vert 0.15$} (B)
        (A) edge [draw=edgecol,in=120,out=60,loop,left] node [xshift=-.1in,yshift=-.1in]{$\sigma_0\vert 0.85$} (A)
        (B) edge [draw=edgecol,in=60,out=120,loop,right] node [xshift=.1in,yshift=-.1in]{$\sigma_0\vert 0.25$} (B)
            edge   [draw=edgecol,bend left]           node {$\sigma_1\vert 0.75$} (A);
\end{tikzpicture}
}
\newcommand{\Gtwo}{%
\begin{tikzpicture}[->,>=stealth',shorten >=1pt,auto,node distance=1.2cm,
                    semithick,scale=\SCALE,font=\bf\fontsize{6}{6}\selectfont]
  \tikzstyle{every state}=[fill=nodecolb,draw=nodeedgeb,text=black,minimum size=3, text width=2,scale=\SCALE,fshadow]
  \node[state] (A)         []           {$\mspace{-6mu}q_0$};
  \node[state]         (B) [right of=A] {$\mspace{-6mu}q_1$};
  \path (A) edge   [draw=edgecol,bend left]           node {$\sigma_1\vert 0.15$} (B)
        (A) edge [draw=edgecol,in=120,out=60,loop,left] node [xshift=-.1in,yshift=-.1in]{$\sigma_0\vert 0.85$} (A)
        (B) edge [draw=edgecol,in=60,out=120,loop,right] node [xshift=.1in,yshift=-.1in]{$\sigma_1\vert 0.75$} (B)
            edge   [draw=edgecol,bend left]           node {$\sigma_0\vert 0.25$} (A);
\end{tikzpicture}
}

\begin{tikzpicture}[scale=1.5,font=\bf \sffamily \fontsize{6}{6}\selectfont]

\node [] (A)  {\Gtwo};
\node [] (aa) at (A.south)  {Synchronizable};

\node [anchor=west] (B) at ([xshift=.05in]A.east)  {\Gone};

\node [] (bb) at (B.south)  {Non-synchronizable};

\end{tikzpicture}
\vspace{-5pt}

\captionN{{\bf Synchronizable and non-synchronizable machines.} Identifying contexts is a key step in estimating the entropy rate of stochastic signals sources; and for PFSA generators, this translates to a state-synchronization problem. However, not all PFSAs are synchronizable, $e.g.$, while the top machine is synchronizable, the bottom one is not. Note that  a history of just one symbol suffices to determine the current state in the synchronizable machine (top), while no finite history can do the same in the non-synchronizable machine (bottom). However, we show that a $\epsilon$-synchronizable string always exists (Theorem~\ref{thmepssynchro}).
}\label{figsync} 
\vspace{-15pt}

\end{figure}

\begin{thm}[$\epsilon$-Synchronization of Probabilistic Automata]\label{thmepssynchro}
 For any QSP $\mathcal{H}$ over  $\Sigma$, the PFSA  $\mathcal{P}_\mathcal{H}$ satisfies: 
\cgather{
\forall \epsilon' > 0, \exists x \in \Sigma^\star, \exists\bvec  \in \mathcal{E},  \vert \vert \wp_x -\bvec \vert \vert_\infty \leqq \epsilon'\label{eqsync}
}
\end{thm}
%

\begin{proof}
 We show that all PFSA are at least approximately synchronizable~\cite{BICP99,Ito84}, which is not true for deterministic automata. If the graph of $\mathcal{P}_\mathcal{H}$ ($i.e.$, the deterministic automaton obtained by removing the arc probabilities) is synchronizable, then Eq.~\eqref{eqsync} trivially holds true for $\epsilon' = 0$ for any synchronizing string $x$. Thus, we assume   the graph of $\mathcal{P}_\mathcal{H}$ to be non-synchronizable. From definition of non-synchronizability, it follows:
\cgather{
\forall q_i,q_j \in Q, \textrm{with } q_i \neq q_j,  \forall x \in \Sigma^\star, \delta(q_i,x) \neq \delta(q_j,x)
}
If the PFSA has a single state, then every string satisfies the condition in Eq.~\eqref{eqsync}. Hence, we assume that the PFSA has more than one state.
Now if we have:
\cgather{
\forall x \in \Sigma^\star, \frac{Pr(x'x)}{Pr(x')} = \frac{Pr(x''x)}{Pr(x'')} \textrm{ where } [x']= q_i, [x'']=q_j
}
then, by the Definition~\ref{defnerode} , we have a contradiction $q_i =q_j$. Hence  $\exists x_0$ such that 
\cgather{
\frac{Pr(x'x_0)}{Pr(x')} \neq \frac{Pr(x''x_0)}{Pr(x'')} \textrm{ where } [x']= q_i, [x'']=q_j \\
\mathrm{Since: } 
\sum_{x \in \Sigma^\star} \frac{Pr(x'x)}{Pr(x')}=1, \textrm{ for any } x' \textrm{ where }  [x']= q_i
}
we conclude without loss of generality  $\forall q_i, q_j \in Q$, with $q_i \neq q_j$:
\cgathers{
\exists x^{ij} \in \Sigma^\star, \frac{Pr(x'x^{ij})}{Pr(x')} > \frac{Pr(x''x^{ij})}{Pr(x'')} \textrm{ where } [x']= q_i, [x'']=q_j
}
It follows from induction that if we start with a distribution $\wp$ on $Q$ such that $\wp_i = \wp_j = 0.5$, then for any $\epsilon' > 0$ we can construct a finite string $x^{ij}_0$ such that if $\delta(q_i,x^{ij}_0) = q_r, \delta(q_j,x^{ij}_0) = q_s$, then for the new distribution $\wp'$ after execution of $x^{ij}_0$ will satisfy
$\wp_s' > 1-\epsilon'$. Recalling that $\mathcal{P}_\mathcal{H}$ is strongly connected, we note that,  for any $q_t \in Q$, there exists a string $y \in \Sigma^\star$, such that $\delta(q_s,y) = q_t$. Setting $x^{i,j\rightarrow t}_\star = x^{ij}_0 y$, we can ensure that the distribution $\wp''$ obtained after execution of $x^{ij}_\star$ satisfies $\wp_t'' > 1-\epsilon'$ for any $q_t$ of our choice. For arbitrary initial distributions $\wp^A$ on $Q$, we must consider contributions arising from simultaneously executing $x^{i,j\rightarrow t}_\star$ from states other than just $q_i$ and $q_j$. Nevertheless, it is easy to see that executing $x^{i,j\rightarrow t}_\star$ implies that in the new distribution $\wp^{A'}$, we have
$ \wp^{A'}_t > \wp^A_i +\wp^A_j -\epsilon'$. It  follows that executing  the string $x^{1,2\rightarrow \vert Q\vert}x^{3,4\rightarrow \vert Q\vert} \cdots x^{n-1,n\rightarrow \vert Q\vert}$, where
\cgather{
n = \begin{cases}
  \vert Q\vert & \textrm{if $\vert Q\vert$ is even}\\
\vert Q \vert -1 & \textrm{otherwise}   
    \end{cases}
}
would result in a final distribution $\wp^{A''}$ which satisfies $\wp^{A''}_{\vert Q \vert } > 1 - \frac{1}{2}n \epsilon'$.
Appropriate scaling of $\epsilon'$ then completes the proof.
\end{proof}
Theorem~\ref{thmepssynchro} induces the notion of  $\epsilon$-synchronizing strings, and guarantees their existence for arbitrary PFSA.
\begin{defn}[$\epsilon$-synchronizing Strings]\label{defepsilonsynchro}
  A  string $x\in \Sigma^\star$ is $\epsilon$-synchronizing for a PFSA if:
\cgather{
\exists\bvec  \in \mathcal{E}, \vert \vert \wp_x -\bvec  \vert \vert_\infty \leqq \epsilon
}
\end{defn}

Theorem~\ref{thmepssynchro} is an existential result, and  does not yield an algorithm for computing synchronizing strings (See Theorem~\ref{thmderivheap}). 
We may  estimate an asymptotic upper bound on such a search.

\begin{cor}[To Theorem~\ref{thmepssynchro}]\label{corsynchrodepth}
At most $O(1/\epsilon)$ strings from the lexicographically ordered set of all strings over the given alphabet need to be analyzed to find an $\epsilon$-synchronizing string.
\end{cor}
\begin{proof}
 Theorem~\ref{thmepssynchro} works by  multiplying entries from the $\Pitilde$ matrix, which cannot be all identical (otherwise the states would collapse). Let the minimum difference between two unequal entries be $\eta$. Then, following the construction in Theorem~\ref{thmepssynchro},  the length $\ell$ of the synchronizing string, up to linear scaling, satisfies: $\eta^\ell = O(\epsilon)$, implying $\ell = O(log(1/\epsilon)$. Hence, the number of strings to be analyzed  is at most all strings of length $\ell$,  where
$ \vert \Sigma\vert^\ell = \vert \Sigma\vert^{O(log(1/\epsilon)} = O(1/\epsilon)$.
\end{proof}
%
\subsection{Symbolic Derivatives}
Computation of $\epsilon$-synchronizing strings requires the notion of symbolic derivatives. Note that, PFSA states  are not  observable; we observe  symbols generated from hidden states. A symbolic derivative at a given string  specifies the distribution of the next symbol over the alphabet.

\begin{notn}
We denote the set of probability distributions 
over a finite  set of cardinality $k$ as $\mathscr{D}(k)$. 
\end{notn}

\begin{defn}[Symbolic Count Function]\label{defcount}
 For a string $s$ over  $\Sigma$, the count function $\#^s: \Sigma^\star \rightarrow \mathbb{N}\cup \{0\}$,  counts the number of times a particular substring occurs in $s$. The count is overlapping, $i.e.$, in a string $s=0001$, we count the number of occurrences of $00$s as $\underline{00}01$ and $0\underline{00}1$, implying $\#^s 00 =2$.
\end{defn}

\begin{defn}[Symbolic Derivative]\label{defsymderivative}
 For a string $s$  generated by a QSP over $\Sigma$, the symbolic derivative  $\phi^s:\Sigma^\star \rightarrow \mathscr{D}(\vert \Sigma\vert -1)$ is defined:
\vspace{-5pt}
\cgather{
\phi^s(x) \big \vert_i = \frac{\#^s x\sigma_i}{\sum_{\sigma_i \in \Sigma }\#^s x\sigma_i}
}
Thus,  $\forall x \in \Sigma^\star, \phi^s(x)$ is a probability distribution over $\Sigma$. $\phi^s(x)$ is referred to as the symbolic derivative at $x$.
\end{defn}

Note that  $\forall q_i \in Q$, $\pitilde$  induces a  probability distribution over $\Sigma$ as  $[\pitilde(q_i,\sigma_1), \cdots , \pitilde(q_i,\sigma_{\vert \Sigma \vert})]$. We denote this as $\pitilde(q_i,\cdot)$.

We next show that the symbolic derivative at $x$ can be used to estimate this distribution for $q_i = [x]$, provided $x$ is $\epsilon$-synchronizing.
\begin{thm}[$\epsilon$-Convergence]\label{thmsymderiv} If $x \in \Sigma^\star$ is $\epsilon$-synchronizing, then:
 \cgather{
\forall \epsilon > 0,  \lim_{\vert s \vert \rightarrow \infty}\vert \vert \phi^s(x) -\pitilde([x],\cdot)\vert \vert_\infty \leqq_{a.s} \epsilon\label{eqsync2} 
}
\end{thm}
\begin{proof}
 We use the Glivenko-Cantelli theorem~\cite{Fl70} on  uniform convergence of empirical distributions. Since $x$ is $\epsilon$-synchronizing:
\cgather{
\forall \epsilon > 0, \exists\bvec  \in \mathcal{E}, \vert \vert \wp_x -\bvec  \vert \vert_\infty \leqq \epsilon
}
Recall that $\mathcal{E}=\big \{ e^i \in [0 \ 1]^{\vert Q \vert } , i= 1,\cdots, \vert Q \vert \big \}$  denotes  the set of distributions over $Q$ satisfying:
\cgather{
e^i\vert_j = \begin{cases}
              1, & \textrm{if } i=j\\
              0, & \textrm{otherwise}
             \end{cases}
}
Let  $x$ $\epsilon$-synchronize to $q \in Q$. Thus, when we encounter $x$ while 
reading  $s$, we are guaranteed to be distributed over  $Q$ as $\wp_x$, where:
\cgather{
 \vert \vert \wp_x -\bvec  \vert \vert_\infty \leqq \epsilon
\Rightarrow \wp_x = \alpha \bvec +(1-\alpha) u
}
where $ \alpha \in [0,1]$, $\alpha \geqq 1 - \epsilon$, and $u$ is an unknown distribution over $Q$. Defining $A_\alpha = \alpha \pitilde(q,\cdot) + (1-\alpha) \sum_{j=1}^{\vert Q\vert}u_j \pitilde(q_j,\cdot)
$, we note that $\phi^s(x)$ is an empirical distribution for $A_\alpha$, implying:
 \caligns{
 &\lim_{\vert s \vert \rightarrow \infty}\vert \vert \phi^s(x) - \pitilde(q,\cdot) \vert \vert_\infty 
 =\lim_{\vert s \vert \rightarrow \infty}
\vert \vert \phi^s(x) - A_\alpha + A_\alpha 
- \pitilde(q,\cdot) \vert \vert_\infty \\
 & \leqq \overbrace{\lim_{\vert s \vert \rightarrow \infty}   \vert \vert \phi^s(x) -A_\alpha \vert\vert_\infty }^{\textrm{\scriptsize  a.s. $0$ by Glivenko-Cantelli}}  +  \lim_{\vert s \vert \rightarrow \infty}\vert \vert  A_\alpha - \pitilde(q,\cdot) \vert \vert_\infty \\
&\leqq_{a.s}  (1-\alpha) \left ( \vert \vert\pitilde(q,\cdot) - u \vert\vert_\infty  \right ) \leqq_{a.s}  \epsilon
 }
This completes the proof.
\end{proof}
\subsection{Computation of  $\epsilon$-synchronizing Strings}
Next we describe identification of  $\epsilon$-synchronizing strings given a sufficiently long observed string ($i.e.$ a sample path) $s$. Theorem~\ref{thmepssynchro}  guarantees existence, and Corollary~\ref{corsynchrodepth} establishes that $O(1/\epsilon)$ substrings need to be analyzed till we encounter an $\epsilon$-synchronizing string.
%
These  do not provide an executable algorithm, which arises from  
an inspection of the geometric structure of the set of probability vectors over $\Sigma$, obtained by constructing $\phi^s(x)$ for different choices of the candidate string $x$.


\begin{defn}[Derivative Heap]\label{defderivheap}
  Given a string $s$ generated by a QSP, a derivative heap $\mathcal{D}^s: 2^{\Sigma^\star} \rightarrow \mathscr{D}(\vert \Sigma \vert -1)$ is the set of probability distributions over $\Sigma$ calculated for a  subset of strings  $L \subset \Sigma^\star$ as:
\cgather{
\mathcal{D}^s(L) = \big \{ \phi^s(x): x \in L \subset \Sigma^\star\big \}
}
\end{defn}
\begin{lem}[Limiting Geometry]\label{lemlimderiv}
 Let us define:
\cgather{
\mathcal{D}_\infty = \lim_{\vert s \vert \rightarrow \infty }\lim_{L \rightarrow \Sigma^\star} \mathcal{D}^s(L)
}
If $\mathscr{U}_\infty$ is the convex hull of $\mathcal{D}_\infty$, and $u$ is a vertex of $\mathscr{U}_\infty$, then 
\cgather{
\exists q \in Q, \textrm{such that } u=\pitilde(q,\cdot)
}
\end{lem}

\begin{proof}
 Recalling Theorem~\ref{thmsymderiv}, the result follows from noting that any element of $\mathcal{D}_\infty$ is a convex combination of elements from the set $\{\pitilde(q_1,\cdot), \cdots , \pitilde(q_{\vert Q\vert},\cdot)  \}$.
\end{proof}

Lemma~\ref{lemlimderiv} does not claim that the number of vertices of the convex hull of $\mathds{D}_\infty$ equals the number of states, but that every vertex  corresponds to a state.  We cannot generate $\mathcal{D}_\infty$  since we 
have a finite observed string $s$, and we can  calculate $\phi^s(x)$ for a finite number of  $x$. Instead, we show that  choosing a string corresponding to the vertex of the convex hull of the  heap, constructed by considering $O(1/\epsilon)$ strings, gives us an $\epsilon$-synchronizing string with high probability.

\begin{thm}[Derivative Heap Approx.]\label{thmderivheap}
 For  $s$ generated by a QSP, let $\mathcal{D}^s(L)$ be computed with $L=\Sigma^{O(log(1/\epsilon))}$. If for  $x_0 \in \Sigma^{O(log(1/\epsilon))}$,  $\phi^s(x_0)$  is a vertex of the convex hull of $\mathcal{D}^s(L)$, then 
\cgather{
Prob(\textrm{$x_0$ is not $\epsilon$-synchronizing}) \leqq e^{-\vert s \vert \epsilon p_0}
}
 where $p_0$ is the probability of encountering  $x_0$ in $s$.
\end{thm}

%
\begin{proof}
The result follows from Sanov's Theorem~\cite{Cs84} for convex set of probability distributions. If $\vert s\vert \rightarrow \infty$, then $x_0$ is guaranteed to be $\epsilon$-synchronizing (Theorem~\ref{thmepssynchro}, and Corollary~\ref{corsynchrodepth}). 
Denoting the number of times we encounter $x_0$ in $s$ as $n(\vert s \vert)$, and since $\mathcal{D}_\infty$ is a convex set of distributions (allowing us to drop the polynomial factor in Sanov's bound), we apply Sanov's Theorem to the case of finite $s$:
\cgather{
Prob\Big (KL\big (\phi^s(x_0) \big \vert \big \vert \wp_{x_0}\Pitilde\big ) > \epsilon \Big ) \leqq e^{-n(\vert s \vert) \epsilon} 
}
where $KL(\cdot \vert \vert \cdot)$ is the Kullback-Leibler divergence~\cite{leh2005}.
From the bound~\cite{tsy09}:
\cgather{
\frac{1}{4}\vert \vert  \phi^s(x_0) - \wp_{x_0}\Pitilde \vert \vert_\infty^2 \leqq  KL\big (\phi^s(x_0) \big \vert \big\vert \wp_{x_0}\Pitilde\big )
}
and  $n(\vert s \vert) \rightarrow \vert s \vert p_0$, where $p_0 > 0$ is the stationary probability of encountering $x_0$ in $s$, we conclude:
\cgather{
Prob\big (\vert \vert  \phi^s(x_0) - \wp_{x_0}\Pitilde \vert \vert_\infty > \epsilon \big ) \leqq 2 e^{-\frac{1}{2}\vert s \vert  \epsilon p_0} 
 \label{eqconfi}
}
 which completes the proof.
\end{proof}
\section{Entropy Rate for PFSA-generated Processes}
Given the  PFSA model, the entropy rate is easily computable.
\begin{thm}[Entropy Rate For PFSA]\label{thmentropyratepfsa}
The entropy rate $H(G)$, in bits,  for the  QSP generated by  a PFSA $G=(Q,\Sigma,\delta,\pitilde)$ is given by:
\cgather{\label{eqnEPFSA}
H(G) = \sum_{i=1}^{\vert Q\vert} \wp_\lambda\big \vert_i \sum_{\sigma_j \in \Sigma} \pitilde(q_i,\sigma_j) \log\pitilde(q_i,\sigma_j)
}
where the base of the logarithms is $2$.
\end{thm}
\begin{proof}
Denote the QSP generated by $G$ as  $\mathcal{X}= \{X_i\}$. Using the chain rule (See Eq.~\eqref{eqchrule}), we have:
\cgather{
H(G) = \lim_{n \rightarrow \infty} \frac{1}{n} \sum_{i=1}^n H(X_i \vert X_{i-1},\cdots,X_1)
}
Since $G$ is always at some state $q \in Q$, we conclude that for any $i$:
\cgathers{
 H(X_i \vert X_{i-1},\cdots,X_1) \in \left \{ \sum_{\sigma_j \in \Sigma} \pitilde(q,\sigma_j) \log\pitilde(q,\sigma_j) : q \in Q  \right \}
}
Furthermore, since $G$ is strongly connected, and therefore has a unique stationary distribution $\wp_\lambda$~\cite{St97}, the number of times state $q_i$ occurs approaches $n\wp_\lambda \big\vert_i$ as $n \rightarrow \infty$. This completes the proof.
\end{proof}
If the underlying PFSA model is not available, and we have only  a symbolic stream generated by a QSP, then  Eq.~\eqref{eqnEPFSA} cannot be directly employed to estimate the entropy rate. In that case, one possibility is to  first infer the hidden PFSA  using the algorithm reported in \cite{CL12g}, and then estimate the entropy rate from Eq.~\eqref{eqnEPFSA}. However, if we are only interested in the latter, then we do not need to infer the complete generative model; and there exists a more parsimonious approach to estimate the entropy rate directly. 

First, we need a lemma which bounds the deviation in entropy for deviations in the probability distribution in the discrete case.
\begin{lem}[Bound on Entropy Deviation]\label{lementropydev}
For probability distributions $p,q$ on a finite set $\Sigma$, we have for all $\epsilon \in (0,1)$,
\cgather{
\vert \vert p -q \vert \vert_\infty \leqq \epsilon \Rightarrow \notag \\\vert H(p) - H(q) \vert <  \epsilon' \log \frac{\vert \Sigma \vert-1}{\epsilon'} + (1-\epsilon') \log \frac{1}{1-\epsilon'} \notag\\
\textrm{where } \epsilon' = \left \{ \begin{array}{cl}
\epsilon & \textrm{if } \epsilon \leqq 1 / 2 \\
1 - \epsilon & \textrm{otherwise} 
\end{array}\right. \notag
}
where $H(p),H(q)$ are  entropies for distributions  $p,q$ respectively.
\end{lem}

\begin{proof}
We have from definition:
\calign{
 H(p) - H(q) &  =  \sum_i p_i \log \frac{1}{p_i} - \sum_i q_i \log \frac{1}{q_i} \notag 
}
We note that the function $f(x) = x \log \frac{1}{x}$ satisfies:
\cgather{
\delta f = \left ( \log \frac{1}{x} - \frac{1}{\ln 2} \right )  \delta x
}
 implying that perturbations of $x$ cause maximum change in $f$, when $x$ is in the neighborhood of $0$, which in turn implies that deviation in entropy for a perturbed distribution $p$ is the maximized  when:
\cgather{
 p \rightarrow p^\star=\begin{pmatrix}
0 & \cdots & 0 & 1
\end{pmatrix} \textrm{upto permutations}
}
Since, $\vert \vert p - q \vert \vert_\infty \leqq \epsilon$, the perturbed distribution $q$ from $p=p^\star$ is non-unique. We claim (Claim A), that the perturbed distribution resulting in maximum entropy deviation, is given by:
\cgather{\label{eqCL}
q^\star = \begin{pmatrix}
\frac{\epsilon'}{\vert \Sigma -1 \vert } & \cdots & \frac{\epsilon'}{\vert \Sigma -1 \vert } & 1- \epsilon'
\end{pmatrix} \textrm{upto permutations} \\
\textrm{where } \epsilon' = \left \{ \begin{array}{cl}
\epsilon & \textrm{if } \epsilon \leqq 1 / 2 \\
1 - \epsilon & \textrm{otherwise} 
\end{array}\right.
}
To establish this claim, we first note that $q^\star$ satisfies the  constraints:
\cgather{
 \forall i \  q^\star_i > 0, \  \sum_i q^\star_i = 1, \  \vert \vert p^\star - q^\star \vert \vert_\infty = \epsilon
}
Let $q'$ be a perturbation of $q^\star$, defined as:
\cgather{
q'_i = q^\star_i + a_i , \textrm{ with }  \sum_i a _i = 0 
}
satisfying the constraint:
\cgather{\label{eqC2}
\vert \vert q' -p^\star \vert \vert_\infty \leqq \epsilon
}
Note that the above constraint, and the definition of $q^\star$ implies that:
\cgather{\label{eqC1}
a_{\vert \Sigma \vert} \geqq 0
}
Then we claim that for small perturbations, 
\cgather{
H(q') < H(q^\star)
}
We find differential perturbations in contribution to the  entropy from perturbation of each entry in $q^\star$. For terms $i \in \{1, \cdots , \vert \Sigma \vert-1 \}$, we note that the 
perturbed term is of the form:
\cgather{
g(x) = \frac{\epsilon' + x}{\vert \Sigma \vert -1} \log \frac{\vert \Sigma \vert -1}{\epsilon' + x}\\
\Rightarrow \delta g(0) = \frac{1}{\vert \Sigma \vert -1} \left ( \log \frac{\vert \Sigma \vert -1}{\epsilon'} - \ln 2   \right )  a_i
}
if $a_i$ is small.  And  the $\vert \Sigma \vert$-th term is of the form:
\cgather{
f(x) = (1-\epsilon'+x) \log \frac{1}{1 - \epsilon +x}\\
\Rightarrow \delta f(0) = \left (\log \frac{1}{1-\epsilon'} - \ln 2 \right )  a_{\vert \Sigma \vert}
}
if $a_{\vert \Sigma \vert}$ is small. This implies that  the perturbation of entropy, for small perturbations in the distribution $q^\star$, is given by:
\mltlne{
\delta H(q^\star) = \frac{1}{\vert \Sigma \vert -1} \left ( \log \frac{\vert \Sigma \vert -1}{\epsilon'} - \ln 2   \right ) \sum_{i=1}^{\vert \Sigma \vert -1} a_i \\ +  \left (\log \frac{1}{1-\epsilon'} - \ln 2 \right ) a_{\vert \Sigma \vert}
}
Noting that $\sum_{i=1}^{\vert \Sigma \vert -1} a_i = -a_{\vert \Sigma \vert}$, and setting $b= \vert \Sigma \vert -1$, we have:
\mltlne{
\delta H(q^\star) = \Bigg ( - \frac{1}{b} \left ( \log \frac{b}{\epsilon'} - \ln 2   \right )   +  \left (\log \frac{1}{1-\epsilon'} - \ln 2 \right ) \Bigg )  a_{\vert \Sigma \vert} \\
= \Bigg  ( \underbrace {\left ( \frac{1}{b} - 1 \right  )\ln 2 }_{t_1}  + \underbrace{\log \frac{1}{1-\epsilon'}  - \frac{1}{b}
\log \frac{b}{\epsilon'}}_{t_2} \Bigg ) a_{\vert \Sigma \vert}
}
We note that since $\vert \Sigma \vert \geqq 2$, $t_1 \leqq 0$. Then, since $\epsilon' \leqq 1/ 2$,  we have: 
\cgather{
\log \frac{1}{1-\epsilon'} \leqq 1 \textrm{ with equality for $\epsilon' = 1/ 2$}
}
And we note that $\frac{1}{b}\log \frac{b}{\epsilon'}$ attains its minimum value of $1 / b + 1/b \log b$  at $\epsilon' = 1/2$, implying:
\cgather{
\delta H(q^\star) \leqq \left ( \frac{1}{b} - 1 \right  )  (\ln 2 - 1) \leqq 0
}
This establishes that within the set of admissible perturbed distributions $q'$, from $q^\star$, 
all infinitesimally small perturbations necessarily reduce the entropy, $i.e.$, $H(q^\star)$ attains a locally maximum value.
We note that for all arbitrary admissible  perturbations $q'$ from $q^\star$, $\vert \vert q' - p^\star \vert \vert_\infty \leqq \epsilon $ and definition of $\epsilon'$ implies that each entry in $q'$ is either always in $[0, 1 / 2]$, or in $[1/2, 1]$, and not both.
 Noting that each summand in the calculation of entropy is of the form $x\log x$, which is monotonic in both intervals, we conclude that $H(q^\star$) is indeed the globally maximum entropy within all admissible perturbations $q'$.
 It follows that any perturbation of $q^\star$, satisfying the constraint of Eq.~\eqref{eqC2}, leads to a smaller difference of entropy from $p^\star$, which establishes claim A. 
Noting that:
\cgathers{
\vert H(p^\star) -H(q^\star) \vert = \epsilon' \log \frac{\vert \Sigma \vert-1}{\epsilon'} + (1-\epsilon') \log \frac{1}{1-\epsilon'}
}
completes the proof.
\end{proof}
This  bound on entropy deviation for $\infty$-norm bounded deviations in   distribution will be important in the sequel. We denote this as the generalized binary entropy function $\mathds{B}(\epsilon, \vert \Sigma \vert )$.
\begin{defn}[Generalized Binary Entropy Function]\label{defbef}
\cgather{
\mathds{B}(\epsilon, \vert \Sigma \vert ) = \epsilon' \log \frac{\vert \Sigma \vert-1}{\epsilon'} + (1-\epsilon') \log \frac{1}{1-\epsilon'}\\
\textrm{where } \epsilon' = \left \{ \begin{array}{cl}
\epsilon & \textrm{if } \epsilon \leqq 1 / 2 \\
1 - \epsilon & \textrm{otherwise} 
\end{array}\right. \notag
}
\end{defn}

\begin{cor}[To Lemma~\ref{lementropydev}]
Given a symbol stream generated by a PFSA  $G=(Q,\Sigma,\delta,\pitilde)$, and an $\epsilon$-synchronizing string $x_0$, we have:
\cgathers{
\left \vert \lim_{n\rightarrow \infty} \frac{1}{\vert \Sigma^n_+ \vert} \sum_{x \in \Sigma^n_+} \lim_{\vert s \vert \rightarrow \infty}H(\phi^s(x_0 x)) - H(G) \right \vert < \mathds{B}(\epsilon, \vert \Sigma \vert )
}
\end{cor}
\begin{proof}
We first establish the following claim (Claim A):  $x_0$ is $\epsilon$-synchronizing implies that  any right extension $x_0 x$  is also $\epsilon$-synchronizing (where $x \in \Sigma^\star$).
To see this, note that  $x_0$ is $\epsilon$-synchronizing implies $\exists\bvec  \in \mathcal{E}$ with:
\cgather{
\wp_{x_0} = \alpha \bvec + (1-\alpha) u, \textrm{with } \alpha \in [0,1], \alpha \geqq 1 - \epsilon 
}
where $u$ is an unknown distribution over  $Q$.
 It follows that: $\forall \sigma \in \Sigma$, 
\cgather{
\wp_{x_0 \sigma} = \frac{1}{\vert \vert \wp_{x_0}\Gamma_\sigma   \vert \vert_1} \left (\alpha \bvec \Gamma_\sigma + (1-\alpha) u \Gamma_\sigma \right )
}
Now, $\forall \sigma \in \Sigma$, there is an unique $\bvec' \in \mathcal{E}$, such that $\bvec' = \bvec \Gamma_\sigma$, and since $\vert \vert \wp_{x_0}\Gamma_\sigma   \vert \vert_1 \leqq 1$, it follows that:
  $\forall \sigma \in \Sigma$, $\exists \bvec' \in \mathcal{E}$, such that:
\cgathers{
\wp_{x_0\sigma} =\alpha' \bvec' + \textrm{additional terms } , \textrm{with } \alpha' \in [0,1], \alpha' \geqq 1 - \epsilon 
}
By straightforward induction, we conclude that:
\cgather{
\forall x \in \Sigma^\star, \exists \bvec(x) \in \mathcal{E}, \textrm{such that }  \vert \vert \wp_{x_0x} - \bvec(x) \vert \vert_\infty \leqq \epsilon
}
which establishes Claim A.


Next we claim (Claim B) that $H(G)$ can be written as:
\cgather{
H(G) =  \lim_{n\rightarrow \infty} \frac{1}{\vert \Sigma^n_+ \vert} \sum_{x \in \Sigma^n_+} H\left (\pitilde([x],\cdot)\right )
}
To see this, note that the PFSA $G$ is strongly connected with  a unique stationary distribution $\wp_\lambda$, and Theorem~\ref{thmentropyratepfsa} implies:
\cgather{
H(G) =  \sum_{i=1}^{\vert Q\vert} \wp_\lambda\big \vert_i  H(\pitilde(q_i,\cdot))
}
Set the initial state of $G$ to be $q \in Q$, where $x_0$ $\epsilon$-synchronizes to $q$.
For any $n$, and each $x \in \Sigma^n_+$, $[x]$ is the equivalence class corresponding to some $q_i \in Q$. Let the number of times $[x]$ corresponds to $q_i$, for $x \in \Sigma^n_+$, be $n_i$. Then, uniqueness of $\wp_\lambda$ implies that $\lim_{n \rightarrow \infty} n_i /n = \wp_\lambda \vert_i$, which implies:
\cgather{
\sum_{i=1}^{\vert Q\vert} \wp_\lambda\big \vert_i  H(\pitilde(q_i,\cdot)) = \lim_{n\rightarrow \infty} \frac{1}{\vert \Sigma^n_+ \vert} \sum_{x \in \Sigma^n_+} H\left (\pitilde([x],\cdot)\right )
}
establishing Claim B. Thus, we can write:
\cgather{
\left \vert \lim_{n\rightarrow \infty} \frac{1}{\vert \Sigma^n_+ \vert} \sum_{x \in \Sigma^n_+} \lim_{\vert s \vert \rightarrow \infty}H(\phi^s(x_0 x)) - H(G) \right\vert\\
= \left \vert \lim_{n\rightarrow \infty} \frac{1}{\vert \Sigma^n_+ \vert} \sum_{x \in \Sigma^n_+} \left \{ \lim_{\vert s \vert \rightarrow \infty}H(\phi^s(x_0 x)) - H(\pitilde([x],\cdot)) \right \} \right\vert\\
\leqq \lim_{n\rightarrow \infty} \frac{1}{\vert \Sigma^n_+ \vert} \sum_{x \in \Sigma^n_+} \left \vert  \lim_{\vert s \vert \rightarrow \infty}H(\phi^s(x_0 x)) - H(\pitilde([x],\cdot))  \right\vert
}
We note that Claim A implies that $x_0x$ $\epsilon$-synchronizes to $[x]$ in $G$, which then implies from  Theorem~\ref{thmsymderiv}, and Lemma~\ref{lementropydev}:
\cgather{
 \lim_{n\rightarrow \infty} \frac{1}{\vert \Sigma^n_+ \vert} \sum_{x \in \Sigma^n_+} \left \vert  \lim_{\vert s \vert \rightarrow \infty}H(\phi^s(x_0 x)) - H(\pitilde([x],\cdot))  \right\vert\\
< \lim_{n\rightarrow \infty} \frac{1}{\vert \Sigma^n_+ \vert} \sum_{x \in \Sigma^n_+} \mathds{B}(\epsilon, \vert \Sigma \vert ) 
}
which completes the proof.
\end{proof}
Next we modify the Dvoretzky-Kiefer-Wolfowitz inequality, to be applicable to the case where the number of samples drawn is itself a random variable.
\begin{lem}[DKW-bound for symbolic derivatives] \label{lemB1} For a string $s$ generated by a PFSA, and a given $\epsilon$-synchronizing string $x_0$:
\cgathers{\forall x \in \Sigma^\star \textrm{ such that  $x_0x$ occurs in  $s$ with probability $\zeta > 0$ },\\
Pr\left  (\left \vert \left \vert \phi^s(x_0x)  -  \lim_{\mathclap{\vert s' \vert \rightarrow \infty}}\phi^{s'}(x_0 x)  \right \vert \right \vert_{\smash{\infty}}  > \epsilon\right ) < 8(1+\frac{1}{e}) e^{-\vert s \vert \zeta \frac{\epsilon^2}{1+\epsilon^2}}
}

\end{lem}
\begin{proof}
We note that $\phi^s(x_0x)$ is an empirical distribution with the limiting distribution given by $ \lim_{\vert s' \vert \rightarrow \infty}\phi^{s'}(x_0 x) \triangleq \phi^\star$. Using the DKW inequality~\cite{massart1990}, and denoting the number of occurrences of $x_0x$ in $s$ with the random variable $N_{x_0x}$, we have:
\mltlne{
Pr\left  (  \left \{ \left \vert \left \vert \phi^s(x_0x)  -  \lim_{\vert s' \vert \rightarrow \infty}\phi^{s'}(x_0 x)  \right \vert \right \vert_\infty > \epsilon \right \}  \bigwedge \{N_{x_0 x} = n'  \} \right )\\ \shoveright{\leqq 2 e^{-2\epsilon^2 n'}Pr\left ( \{N_{x_0 x} = n'  \}  \right )
}\\
\shoveleft{\Rightarrow Pr\left  (   \left \vert \left \vert \phi^s(x_0x)  -  \lim_{\vert s' \vert \rightarrow \infty}\phi^{s'}(x_0 x)  \right \vert \right \vert_\infty > \epsilon \right )} \\ \leqq \sum_{n' \in \mathbb{N}} 2 e^{-2\epsilon^2 n'}Pr\left ( \{N_{x_0 x} = n'  \}  \right )
}
We partition $\mathbb{N}$ into disjoint sets $U_r$ and $V_r =  \mathbb{N} \setminus U_r$, parametrized by $r > 0$, where:
\cgather{
U_r = \Bigg [ \bigg \lfloor \vert s \vert \zeta (1- r) \bigg \rfloor , \bigg \lceil \vert s \vert \zeta (1+ r) \bigg \rceil \Bigg ] 
}
Using Chernoff bounds for the probability of $n' \in V_r$, we have:
\mltlne{
 Pr\left  (   \left \vert \left \vert \phi^s(x_0x)  -  \lim_{\vert s' \vert \rightarrow \infty}\phi^{s'}(x_0 x)  \right \vert \right \vert_\infty > \epsilon \right ) 
\\\shoveleft{\leqq \sum_{n' \in U_r} 2 e^{-2\epsilon^2 n'}Pr\left ( \{N_{x_0 x} = n'  \}  \right )} 
\\ \shoveright {
+ \sum_{n' \in V_r} 2 e^{-2\epsilon^2 n'}Pr\left ( \{N_{x_0 x} = n'  \}  \right )
}
\\\shoveleft{\leqq
\bigg ( \big \lceil 2\vert s \vert \zeta r \big \rceil \times 2e^{-2 \epsilon^2 \vert s \vert \zeta (1-r)} \times 1 \bigg ) + \bigg (1 \times 2e^{-\frac{r^2 \vert s \vert \zeta}{2+r}} \bigg )} \\
\leqq 4\big \lceil \vert s \vert \zeta r \big \rceil e^{-2 \epsilon^2 \vert s \vert \zeta (1-r)} +  2e^{-\frac{r^2 \vert s \vert \zeta}{2+r}}
}
Denoting $\vert s \vert \zeta $ as $t$, we have the bound:
\cgather{
\forall r > 0, f(r) = 4 \big \lceil rt \big \rceil e^{-2\epsilon^2t(1-r)}+2e^{-\frac{r^2t}{2+r}}
}
 We note that the two terms are equal if:
\cgather{
2\epsilon^2t(1-r) = \frac{r^2t}{2+r} + \ln(2 \lceil rt \rceil ) 
}
It follows that if we solve for $r$ in terms of $\epsilon$ after dropping the  non-negative log-term, then the first term would be bigger or equal compared to the second. Solving the resulting quadratic, we get:
\cgather{
r < \frac{\epsilon^2}{1+\epsilon^2}
}
A larger value of $r$  makes the first term larger, and the second term smaller; hence we use $r= \frac{\epsilon^2}{1+\epsilon^2}$, leading to the non-tight  bound:
\cgather{
f(r) < 8 \left \lceil \frac{\epsilon^2}{1+\epsilon^2} t \right \rceil e^{-2\frac{\epsilon^2}{1+\epsilon^2}t} < 8(1+\frac{\epsilon^2}{1+\epsilon^2} t) e^{-2\frac{\epsilon^2}{1+\epsilon^2}t}
}
Using the fact that  $\forall y \in \mathbb{R}, 1-y \leqq e^{-y}$, we have:
\cgather{
f(r) < 8e^{-2\frac{\epsilon^2}{1+\epsilon^2}t} + 8 e^{-\frac{\epsilon^2}{1+\epsilon^2}t -1 } < 8\left (1 + \frac{1}{e} \right ) e^{-\frac{\epsilon^2}{1+\epsilon^2}t }
}
which completes the proof.
\end{proof}

\begin{cor}[To Lemma~\ref{lemB1}]\label{corlemB1} For $s$ generated by a PFSA, and an $\epsilon$-synchronizing $x_0$, we have for any $x \in \Sigma^\star$:
\cgathers{
 Pr\left  (  \bigg \vert H(\phi^s(x_0x) )  -  H\left (\lim_{\vert s' \vert \rightarrow \infty}\phi^{s'}(x_0 x)\right )  \bigg \vert  >  \mathds{B}(\epsilon, \vert \Sigma \vert )\right ) \notag \\ <  8\left (1+\frac{1}{e}\right ) e^{-\vert s \vert \zeta \frac{\epsilon^2}{1+\epsilon^2}}}
\end{cor}
\begin{proof}
It follows from Lemma~\ref{lementropydev} and continuity of entropy that 
\mltlne{
  \big \vert \big \vert \phi^s(x_0x)  -  \lim_{\vert s' \vert \rightarrow \infty}\phi^{s'}(x_0 x)  \big \vert \big \vert_\infty \leqq \epsilon \notag \\ \Rightarrow 
 \bigg \vert  H(\phi^s(x_0x) )  -  H\left ( \lim_{\vert s' \vert \rightarrow \infty} \phi^{s'}(x_0 x)\right )  \bigg \vert  \leqq \mathds{B}(\epsilon, \vert \Sigma \vert )\\
\shoveleft{\textrm{Using Lemma~\ref{lemB1}, we have: }}\\
\shoveleft{Pr\left  (\big \vert \big \vert \phi^s(x_0x)  -  \lim_{\mathclap{\vert s' \vert \rightarrow \infty}}\phi^{s'}(x_0 x)  \big \vert \big \vert_{\smash{\infty}}   \leqq  \epsilon\right ) }
\geqq 1-  8(1+\frac{1}{e}) e^{-\vert s \vert \zeta \frac{\epsilon^2}{1+\epsilon^2}}
 \\
\Rightarrow 
Pr\left  (  \bigg \vert  H(\phi^s(x_0x) )  -  H\left ( \lim_{\vert s' \vert \rightarrow \infty}\phi^{s'}(x_0 x)\right )  \bigg \vert  \leqq \mathds{B}(\epsilon, \vert \Sigma \vert )\right )  \\ \geqq 1-  8\left (1+\frac{1}{e}\right ) e^{-\vert s \vert \zeta \frac{\epsilon^2}{1+\epsilon^2}}   
}
which completes the proof.
\end{proof}
%
\begin{thm}[Bound on Entropy Calculation with Finite Samples]\label{thmB1} For any string $x$ generated by a PFSA $G=(Q,\Sigma,\delta,\pitilde)$, and a given $\epsilon$-synchronizing string $x_0$, there exist  $ C_0, C_1$  depending only on the size of the alphabet $\vert \Sigma \vert$, such that, for any independently chosen set of strings  $\ND \subseteqq \Sigma^\star$:
\cgathers{
Pr \Bigg ( \bigg \vert   \frac{1}{\vert\ND \vert} \sum_{\mathclap{\phantom{X} x \in \ND}} H(\phi^s(x_0 x))   - \lim_{n\rightarrow \infty} \frac{1}{\vert \Sigma^{n}_+ \vert} \sum_{x \in \Sigma^{n}_+} \lim_{\mathclap{\phantom{X} \vert s \vert \rightarrow \infty}}H(\phi^s(x_0 x)) \bigg \vert \notag \\ > \mathds{B}(\epsilon, \vert \Sigma \vert ) + \epsilon\Bigg ) \leqq  C_0 \frac{1+\epsilon^2}{\vert s \vert \epsilon^3} + 2 e^{C_1\vert \ND \vert  \epsilon^{2} } 
}
\end{thm}
\begin{proof}
We note that:
\cgather{
A = \bigg \vert   \frac{1}{\vert \ND \vert} \sum_{\mathclap{x \in \ND}} H(\phi^s(x_0 x))   - \lim_{n\rightarrow \infty} \frac{1}{\vert \Sigma^{n}_+ \vert} \sum_{x \in \Sigma^{n}_+} \lim_{\mathclap{\vert s \vert \rightarrow \infty}}H(\phi^s(x_0 x)) \bigg \vert \notag \\
\leqq \bigg \vert   \frac{1}{\vert \ND \vert} \sum_{\mathclap{x \in \ND}} H(\phi^s(x_0 x))   -  \frac{1}{\vert \Sigma^{n}_+ \vert} \sum_{x \in \Sigma^n_+} \lim_{\mathclap{\vert s \vert \rightarrow \infty}}H(\phi^s(x_0 x)) \bigg \vert \notag \\
+  \bigg \vert   \frac{1}{\vert \ND \vert} \sum_{\mathclap{x \in \ND}} \lim_{\mathclap{\phantom{xxx}\vert s \vert \rightarrow \infty}}H(\phi^s(x_0 x))     - \lim_{\mathclap{n\rightarrow \infty\phantom{x}}} \frac{1}{\vert \Sigma^{n}_+ \vert} \sum_{\mathclap{x \in \Sigma^{n}_+\phantom{x}}} \lim_{\mathclap{\phantom{xx}\vert s \vert \rightarrow \infty}}H(\phi^s(x_0 x)) \bigg \vert \notag
}
We denote the two RHS terms as $B$ and $C$, and note:
\cgather{
B \triangleq \bigg \vert   \frac{1}{\vert \ND \vert} \sum_{\mathclap{x \in \ND}} H(\phi^s(x_0 x))   -  \frac{1}{\vert \Sigma^{n}_+ \vert} \sum_{x \in \Sigma^n_+} \lim_{\mathclap{\vert s \vert \rightarrow \infty}}H(\phi^s(x_0 x)) \bigg \vert \notag\\
= \bigg \vert   \sum_{i=1}^{\vert Q\vert} \widetilde{\wp}_i \left (  H(\phi^s(x_0x'))   -  \lim_{\mathclap{\vert s \vert \rightarrow \infty}}H(\phi^s(x_0x')) \right ) \bigg \vert 
\intertext{where $x_0$ $\epsilon$-synchronizes to $q_0\in Q$, $\delta(q_0,x') = q_i$, and $\widetilde{\wp}$ is the empirical estimate of the stationary distribution. Using the bound from Corollary~\ref{corlemB1}: }
Pr(B > \mathds{B}(\epsilon, \vert \Sigma \vert )) <  8\left (1+\frac{1}{e}\right ) \sum_{i=1}^{\vert Q\vert}\widetilde{\wp}_i  e^{-\vert s \vert \widetilde{\wp}_i \frac{\epsilon^2}{1+\epsilon^2}} 
\notag \\
 <  8\left (1+\frac{1}{e}\right ) \frac{(1+\epsilon^2)\vert Q\vert}{e \vert s \vert \epsilon^2 } \label{eqB}
}
For the second RHS term:
\cgather{
C \triangleq  \bigg \vert   \frac{1}{\vert\ND \vert} \sum_{\mathclap{x \in \ND}} \lim_{\mathclap{\phantom{xxx}\vert s \vert \rightarrow \infty}}H(\phi^s(x_0 x))     - \lim_{\mathclap{n\rightarrow \infty\phantom{x}}} \frac{1}{\vert \Sigma^{n}_+ \vert} \sum_{\mathclap{x \in \Sigma^{n}_+\phantom{x}}} \lim_{\mathclap{\phantom{xx}\vert s \vert \rightarrow \infty}}H(\phi^s(x_0 x)) \bigg \vert\notag \\
= \bigg \vert    \sum_{i=1}^{\vert Q\vert} \left ( \widetilde{\wp}_i -  \wp_\lambda \big \vert_i  \right )\lim_{\vert s \vert \rightarrow \infty}  H(\phi^s(x_0 x'))   \bigg \vert \leqq \| \widetilde{\wp} -\wp_\lambda \|_1 \log \vert \Sigma \vert
\intertext{where $x_0$ $\epsilon$-synchronizes to $q_0\in Q$ and $\delta(q_0,x') = q_i$. Using DKW:}
Pr(\| \widetilde{\wp} - \wp_\lambda \|_\infty > \epsilon) \leqq 2e^{-2 \vert\ND\vert \epsilon^2} \notag \\
\Rightarrow Pr (\| \widetilde{\wp} - \wp_\lambda \|_1  \log \vert \Sigma \vert \leqq \epsilon ) > 1-  2e^{-\frac{2}{\log^2 \vert \Sigma \vert} \vert\ND\vert \epsilon^2} \label{eqD} 
}
Using the bounds in Eq.~\eqref{eqB}, and \eqref{eqD}, we get:
\cgathers{
E \triangleq Pr \left  (B+C \leqq \mathds{B}(\epsilon, \vert \Sigma \vert ) + \epsilon \right ) \\ > \left (1 - 8 (1+1/e ) \frac{(1+\epsilon^2)\vert Q\vert}{e \vert s \vert \epsilon^2 }  \right ) \times \left (  1- 2 e^{- \vert \ND \vert  \epsilon^2 \frac{2}{\log^2 \vert \Sigma \vert}   }  \right )
}
Since we are using $\epsilon$-synchronization, it follows that the number of states $\vert Q\vert$ is upper bounded by $(\vert \Sigma \vert -1)/ \epsilon$
which then yields:
\calign{
&E > \left (1 - C_0 \frac{1+\epsilon^2}{\vert s \vert \epsilon^{3}} \right  ) \left (1 - 2 e^{-C_1\vert \ND \vert \epsilon^2 } \right  )   \label{eqCdef}
\intertext{with $
C_0 = (8/e+8/e^2)(\vert \Sigma \vert -1), \textrm{and } C_1 = \frac{2}{\log^2 \vert \Sigma \vert}$}
\Rightarrow &Pr \left  (A \leqq \mathds{B}(\epsilon, \vert \Sigma \vert ) + \epsilon \right ) > 1 - C_0 \frac{1+\epsilon^2}{\vert s \vert \epsilon^{3}} - 2 e^{-C_1\vert \ND \vert \epsilon^{2} } \notag
}
which completes the proof.
\end{proof} 
%
\def\TOL{\textrm{\scshape Count}_\textrm{total}}
\def\Tol{\textrm{\scshape Count}_\textrm{map}}
\def\Num{\textrm{\scshape Num}}
\begin{algorithm}[t]\sffamily \fontsize{7}{8}\selectfont
\SetKw{BRK}{break}
\DontPrintSemicolon
\KwIn{Data sequence $s$ over alphabet $\Sigma$, $\epsilon$,  Confidence level $\alpha$ }
\KwOut{ Entropy rate $\mathbf{h}$, Uncertainty $\mathbf{E}$ at specified confidence level }
\BlankLine

Initialize $h=0$\;
Initialize $\TOL = 0$\;
Initialize $\Tol = \varnothing$  \color{Red4}
 \tcc*[f]{\texttt{ \fontsize{6}{6}\selectfont hashtable with keys as  probability distributions, $\phantom{xxxxxxxxxx.xxxxxxxxxxxxxxx}$  and values as doubles}}\color{black}\;
Set $ C_0 = (8/e+8/e^2)(\vert \Sigma \vert -1)$,   $C_1  = 2 / \log^2 \vert \Sigma \vert $\;
Set $N_{\mathrm{min}} = 10$\color{Red4}\tcc*[f]{ \texttt{ \fontsize{6}{6}\selectfont Any small integer suffices (See Section~\ref{secimpl})}}\color{black}\;
\BlankLine
\rule{.945\columnwidth}{1pt}
\color{DodgerBlue4}
\tcc*[f]{\scriptsize I. $\epsilon$-synchronization String Identification}\;\color{black}
\BlankLine

\ForEach{$x \in \Sigma^{\log(1/\epsilon)}_+$}{
$D[x] \longleftarrow \phi^s(x) $\;
}

 $A \longleftarrow \{x': \textrm{$D[x']$ is on the convex hull of the set of values in hashtable $D$} \}$\;
 $x_0 \longleftarrow \argmax_{x \in A} \#^sx$\color{Red4}\tcc*[f]{ \texttt{ \fontsize{6}{6}\selectfont $\epsilon$-synchronization string}}\color{black} \;
$p_0  \longleftarrow (\#^sx_0 )\vert s \vert $ \color{Red4}\tcc*[f]{ \texttt{ \fontsize{6}{6}\selectfont Occurrence prob. of $\epsilon$-synchronization string}}\color{black}\;
\rule{.945\columnwidth}{1pt}
\color{DodgerBlue4}
\tcc*[f]{\scriptsize II. Entropy Rate Estimatation}\;\color{black}
\BlankLine

Select  $\ND \subset \Sigma^\star $ with length $\ell$  strings drawn with probability $\frac{1}{\vert \Sigma \vert^\ell}$\;
\color{Red4}\tcc*[h]{ \texttt{ \fontsize{6}{6}\selectfont $\vert \ND \vert \sim 10^7 \log^2 \vert \Sigma \vert$ sufficient for negligible  uncertainty contribution}}\color{black}\;
\ForEach{$x \in \ND$}{
\eIf{$\#^s x_0x > N_{\mathrm{min}}$}{
Compute $u \longleftarrow \phi^s(x_0x)$\color{Red4}\tcc*[f]{ \texttt{ \fontsize{6}{6}\selectfont Symbolic derivative at $x_0x$}}\color{black}\;
\eIf{$\exists \textbf{ key } v \in \Tol \textrm{ s.t. }  \| u - v \|_\infty \leqq \epsilon $}{
$\Tol[v] \longleftarrow \Tol [v] +1$\;
}{
Set $\Tol[u] = 1$\;
}
$\TOL \longleftarrow \TOL +1$\;
}{
Delete $x$ from $\ND$\;
}
}

\ForEach{$\mathrm{key} \ v \in \Tol$}{
$\mathbf{h} \longleftarrow \mathbf{h} + \left (\dfrac{\Tol[v]}{\TOL}  \times H(v)\right )$\color{Red4}\tcc*[f]{ \texttt{ \fontsize{6}{6}\selectfont $H(v)$: entropy of $v$}}\color{black}\;
}
\rule{.945\columnwidth}{1pt}
\color{DodgerBlue4}
\tcc*[f]{\scriptsize III. Uncertainty Estimation}\;\color{black}
\BlankLine

$\epsilon_\star \longleftarrow \min \epsilon_0$ satisfying:
$ \alpha+  C_0 \frac{1+\epsilon^2_0}{\vert s \vert \epsilon^{3}_0} + 2 e^{-C_1\vert \ND \vert  \epsilon^{2}_0 } + e^{-\epsilon_0 p_0 \vert s \vert} \leqq  1 $\;
$\mathbf{E} \longleftarrow \epsilon_\star + 2 \mathds{B}(\epsilon_\star, \vert \Sigma \vert)$\;
\KwRet $\mathbf{h}$, $\mathbf{E}$\; 
\captionN{Detailed pseudocode for entropy rate estimation}\label{algo1}
\end{algorithm}

\begin{figure}[t]
\tikzexternalenable 
\centering
\def\HGT{1.5in}
\def\WDT{2.85in}
\def\SCALE{1}
\def\datafile{Figures/alph27.dat}
\begin{tikzpicture}[font=\bf \sffamily \fontsize{6}{6}\selectfont]
 \pgfplotsset{every axis legend/.append style={
at={(0.265,0.625)},
anchor=south}}
\begin{semilogxaxis}[ title={(a) Alphabet Size = $27$}, title style={yshift=-.1in},legend cell align=left,
legend style={ xshift=1in, yshift=-.1in, draw=white, fill= gray, fill opacity=0.2, 
text opacity=1,},
axis line style={black, opacity=0.5,  thick, rounded corners=0pt},
axis on top=false, 
scale=1,grid style={dashed, gray!40},
enlargelimits=false, 
width=\WDT, 
height=\HGT,     
,ymax=1.5, 
 ymin=0.00001,
semithick,grid,
axis background/.style={top color=gray!0,},
xlabel={input length [\# of symbols]},
yticklabel style={xshift=.025in},
ylabel={error [bits / letter]},
ylabel style={yshift=-.2in},
 xlabel style={yshift=.1in},
    scaled x ticks = false,
      x tick label style={/pgf/number format/fixed,
      /pgf/number format/1000 sep = \thinspace 
      }
  ];
%


\addplot[  thick, Red4]table[x index={0}, y expr=(\thisrowno{1}*\SCALE)]
  {\datafile};
\addlegendentry{95\% confidence};

\addplot[  thick, DodgerBlue4]table[x index={2}, y expr=(\thisrowno{3}*\SCALE)]
  {\datafile};
\addlegendentry{99\% confidence};

\node [circle, fill=none, draw=black, inner sep=1pt] (a) at (axis cs:5e6, .9) {};
\draw[gray, thin](a -| current plot begin) -- (a) node [midway,sloped, yshift=0.02in, below, text=black!80, xshift=0.08in, font=\bf \sffamily \fontsize{6}{6}\selectfont] {$0.9$};
\draw[gray, thin]  (a) -- (current plot end -| a) node [midway,sloped, yshift=-0.03in, above, text=black!80, xshift=0.03in, font=\bf \sffamily \fontsize{6}{6}\selectfont] {$5\times 10^6$};

\end{semilogxaxis}
\end{tikzpicture}

\vspace{2pt}

\def\SCALE{1}
\def\datafile{Figures/alph2.dat}
\begin{tikzpicture}[font=\bf \sffamily \fontsize{6}{6}\selectfont]
 \pgfplotsset{every axis legend/.append style={
at={(0.265,0.625)},
anchor=south}}
\begin{semilogxaxis}[ title={(b) Alphabet Size = $2$}, title style={yshift=-.1in},legend cell align=left,
legend style={ xshift=1in, yshift=-0.1in, draw=white, fill= gray, fill opacity=0.2, 
text opacity=1,},
axis line style={black, opacity=0.5,  thick, rounded corners=0pt},
axis on top=false, 
scale=1,grid style={dashed, gray!40},
enlargelimits=false, 
width=\WDT, 
height=\HGT,     
,ymax=.95, 
 ymin=0.00001,
semithick,grid,
axis background/.style={top color=gray!0,},
xlabel={input length [\# of symbols]},
yticklabel style={xshift=.025in},
ylabel={error [bits / letter]},
ylabel style={yshift=-.2in},
 xlabel style={yshift=.1in},
    scaled x ticks = false,
      x tick label style={/pgf/number format/fixed,
      /pgf/number format/1000 sep = \thinspace 
      } 
  ];
%


\addplot[  thick, Red4]table[x index={0}, y expr=(\thisrowno{1}*\SCALE)]
  {\datafile};
\addlegendentry{95\% confidence};

\addplot[  thick, DodgerBlue4]table[x index={2}, y expr=(\thisrowno{3}*\SCALE)]
  {\datafile};
\addlegendentry{99\% confidence};

\node [circle, fill=none, draw=black, inner sep=1pt] (a) at (axis cs:5e6, .22) {};
\draw[gray, thin](a -| current plot begin) -- (a) node [midway,sloped, yshift=-0.02in, above, text=black!80, xshift=0.08in, font=\bf \sffamily \fontsize{6}{6}\selectfont] {$0.22$};
\draw[gray, thin]  (a) -- (current plot end -| a) ;

\end{semilogxaxis}
\end{tikzpicture}
\vspace{-7pt}

\captionN{\textbf{Uncertainty bounds  for different alphabet sizes.} Note for a data length of $5\times 10^6$, we have an uncertainty of $0.9$ bits at $95 \%$ confidence for a $27$ letter alphabet (plate (a)); the corresponding uncertainty for a binary alphabet is $0.22$ bits (plate(b)). }\label{figBND}
\end{figure}
\begin{figure*}[t] 
\tikzexternalenable 
\centering
\input{Figures/figTXT.tex}

\vspace{-7pt}

\captionN{\textbf{Applications.} Plates (a-b): Entropy rate of English text. Shannon's experiment using human subjects  puts the estimate around $1$ bit per letter. We achieve very close estimates. The state-of-the-art plots are replicated from~\cite{shgs96}. Plates (c-d): Entropy rate of sequences generated by a chaotic dynamical system and a binary generating partition. Plates (e-f): Entropy of symbol streams generated by probabilistic automata. Note that even with two states, and a binary alphabet, a non-synchronizable generating process leads to significantly larger errors with the LZ-based approaches.}\label{figTXT}
\end{figure*}
\begin{thm}[Main Theorem]\label{thmB2}
Given a finite string $s$ generated by a PFSA  $G=(Q,\Sigma,\delta,\pitilde)$, and a string  $x_0 \in \Sigma^\star$ satisfying the pre-conditions described in Theorem~\ref{thmderivheap}, 
we have for  any independently chosen set of strings $\ND \subseteqq \Sigma^\star$:
\cgather{
Pr \Bigg ( \bigg \vert   \frac{1}{\vert \ND \vert} \sum_{\mathclap{x \in \ND}} H(\phi^s(x_0 x))   - H(G) \bigg \vert > \epsilon +2\mathds{B}(\epsilon, \vert \Sigma \vert ) \Bigg ) \notag \\ \leqq C_0 \frac{1+\epsilon^2}{\vert s \vert \epsilon^{3}} + 2 e^{-C_1\vert \ND \vert  \epsilon^{2} } + e^{-\epsilon p_0 \vert s \vert} \label{eqBND}
}
where $
C_0=(8/e+8/e^2)(\vert \Sigma \vert -1)$,    $C_1=2 / \log^2 \vert \Sigma \vert$ and $p_0$ is the non-zero occurrence probability of $x_0$ in $s$.
\end{thm}
\begin{proof}
It follows  from Corollary~\ref{lementropydev}, and  Theorem~\ref{thmB1}, that:
\cgather{
T \triangleq Pr \Bigg ( \bigg \vert   \frac{1}{\vert \ND \vert} \sum_{\mathclap{x \in \ND}} H(\phi^s(x_0 x))   - H(G) \bigg \vert \leqq  \epsilon +2\mathds{B}(\epsilon, \vert \Sigma \vert ) \Bigg ) \notag \\ > \left (1- C_0 \frac{1+\epsilon^2}{\vert s \vert \epsilon^{3}} - 2 e^{-C_1\vert\ND \vert) \epsilon^{2} }\right )\times Pr(x_0 \textrm{ is $\epsilon$-synchronizing})  \notag 
\intertext{Assuming $x_0$ satisfies the pre-conditions described in  Theorem~\ref{thmderivheap}:}
T > \left (1- C_0 \frac{1+\epsilon^2}{\vert s \vert \epsilon^{3}} - 2 e^{-C_1\vert \ND \vert \epsilon^{2} }\right )\times \left (1 - e^{\epsilon p_0 \vert s \vert}  \right )  \notag \\
> 1- C_0 \frac{1+\epsilon^2}{\vert s \vert \epsilon^{3}} - 2 e^{-C_1\vert \ND \vert \epsilon^{2} } - e^{\epsilon p_0 \vert s \vert}  
}
which completes the proof.
\end{proof}
\begin{rem}
We note the following:
\begin{itemize}
\item For binary alphabets, we have: $
C_0 \simeq 4.03, C_1 = 2
$.
\item Each term on  the RHS of Eq.~\eqref{eqBND} reflects   a specific contribution:
\cgather{
\underbrace{C_0 \frac{1+\epsilon^2}{\vert s \vert \epsilon^{3}}}_{\mathclap{\text{Data-length Dependence}}} + \overbrace{2 e^{-C_1\vert \ND \vert \epsilon^{2} }}^{\mathclap{\text{Dependence on Summation Depth}}} + \underbrace{e^{-\epsilon p_0 \vert s \vert}}_{\mathclap{\text{Synchronization Error Dependence}}}
}
\end{itemize}
\item Eq.~\eqref{eqBND} bounds the maximum uncertainty at a given confidence level, which depends on the alphabet size. The uncertainty relationships for two alphabet sizes ($2$, $27$) are shown in Figure~\ref{figBND}.
\end{rem}
\begin{cor}[To Theorem~\ref{thmB2}]
As a function of the length of the observed data string $s$, the upper and lower confidence bands for the estimated entropy rate, with any fixed confidence level, 
converge at a rate $\displaystyle O\left (\frac{\log \vert s \vert }{ \vert s \vert^{1/3}}\right )$.
\end{cor}
\begin{proof}
For  a given  confidence level  $k=k_{data}+k_{depth}$, where $k_{data}$  captures the dependence on the data length through the first RHS term in Eq.~\eqref{eqBND}, we get:
\cgather{
\epsilon^3 = \frac{C_0(1+\epsilon^2)}{ \vert s \vert k_{data} } 
\Rightarrow \epsilon < \left ( \frac{2C_0}{ \vert s \vert k_{data} } \right )^{1/3} \label{eq105}
}
The distance between the confidence bands is given by:
\cgather{
\mathcal{B} = 2\epsilon + 4\mathds{B}(\epsilon, \vert \Sigma \vert )
}
and using  Eq.~\eqref{eq105}, along  with the definition of the generalized binary entropy function (Definition~\ref{defbef}), completes the proof.
\end{proof}
\section{Algorithmic Implementation}\label{secimpl}
The algorithmic steps for the proposed entropy rate estimation technique is enumerated in Algorithm~\ref{algo1}. The inputs to the algorithm is the data stream $s$, $\epsilon$, and the confidence level $\alpha$ at which the error  estimate is desired.
Importantly, the size of the set of sampled string $\ND$ is not required to be an input; 
if computational effort is not a concern, then the uncertainty contribution from the term involving $\vert \ND \vert$ (See Eq.~\eqref{eqBND}) can be reduced to negligible levels by using a sample set with  
$\vert \ND \vert \simeq \frac{K}{2\epsilon^2} \log^2 \vert \Sigma \vert  
$, which would result in  uncertainty contribution of $\sim e^{-K}$. Using $\vert \ND \vert \simeq 10^7 \log^2 \vert \Sigma \vert$ is generally sufficient to make this factor negligible; smaller sets may be used under computational constraints, which would  lead to  increased uncertainty in the entropy estimate.

Particularly rare strings may accumulate errors, which is prevented in the implementation by ignoring strings that occur too infrequently (Note $N_\mathrm{min}$ in step 5 and step 15 of Algorithm ~\ref{algo1}).
\subsection{Application to English text, Chaotic systems \& Random walks}
We demonstrate  Algorithm~\ref{algo1} in three different applications. 
Our first application is the estimation of the entropy rate of English text. 
Shannon's experimental approach with human subjects~\cite{shannon1951} suggests that English has an entropy of around one bit per letter. However, the  large alphabet size ($26$ letters + space = $27$), makes it computationally hard to verify this value. We apply our algorithm to relatively small corpora: the King James Bible (KJB) (which has a length $\sim 4 \times 10^6$ letters), and the collected works of Shakespeare (SHK, length $\sim 4.8 \times 10^6$ letters). These  particular examples allow direct  comparison against the results reported in \cite{shgs96}. We obtain entropy rates which are significantly closer to the Shannon estimate (See Figure~\ref{figTXT}):  $1.05  \btl$ for KJB, and $1.25  \btl$  for SHK, while Sch{\"u}rmann $\etal$ obtain the corresponding estimates  to be $1.73 \btl$ and $2.13 \btl$. The authors in \cite{shgs96} were able to improve the SHK estimate to  $1.7 \btl$ using the ``ansatz'' mentioned before;  Algorithm~\ref{algo1} yields an improved estimate without any such assumptions. 

Our second application is entropy estimation of sequences produced by chaotic dynamical systems. We use the same iteration map used in \cite{shgs96}: namely $x_{n+1} = 1 - r x^2_n$, and use a binary generating partition at $x=0$. We analyze the  cases
$r=1.7499$ (Figure~\ref{figTXT}(b)) where it is very strongly intermittent, and $r=1.75$ which is the Pomeau-Manneville intermittency point (Figure~\ref{figTXT}(c)). As before, we converge faster in the non-trivial case, and gets very close to the theoretical entropy given by the positive  Lyapunov exponent due to Pesin's identity~\cite{rissanen83}.

Our third application analyzes sequences generated by finite memory ergodic stationary stochastic processes, modeled directly via probabilistic automata (Figure~\ref{figTXT}(e-f)). Thus, we are looking at generalized random walks, Inspite of being  somewhat more contrived compared to the first two applications, we can gain important insights from this example. Even with two states, and with a binary alphabet, LZ-based approaches may perform significantly worse, particularly for short streams with long range dependencies. We note that the PFSA generator used in Figure~\ref{figTXT}(e) is non-synchronizable, $i.e.$, no finite length of observed history tells us definitively what the current state is. Nevertheless, as we showed in Theorem~\ref{thmepssynchro}, the machine is $\epsilon$-synchronizable; and Algorithm~\ref{algo1} performs quite well, converging to the theoretical value with just under $10^4$ symbols. In contrast, the LZ-compression based algorithm has an error of about $17\%$ even after $3 \times 10^4$ symbols. This is discrepancy in performance disappears if the generating process is synchronizable, $e.g.$, if a finite history tells us precisely what the current state is. Indeed with a synchronizable PFSA in Figure~\ref{figTXT}(f) (here, the last symbol is sufficient to fix the current state), the algorithms have comparable performances.

\section{Summary \& Conclusion}\label{sec6}
We delineate a new algorithm for estimating entropy rates of symbol streams,  generated by  hidden ergodic stationary processes. We establish the correctness of the algorithm by exploiting a connection with the theory of probabilistic automata, and that of finite measures on infinite strings. Importantly, we establish a distribution-free limit theorem. Using established results from non-parametric statistics, we show that  entropy  estimate converges at the rate $O(\log \vert s \vert / \sqrt[3]{ \vert s \vert} )$ as a function of the input data length $\vert s \vert$. In consequence, we are able to derive confidence bounds on the  estimate, and dictate the  worst-case data length required to guarantee a specified error bound at a given confidence level. Finally, we demonstrate  that, in terms of data requirements,  the proposed algorithm has superior performance to competing approaches, at least in the case of the chosen applications. 

}
 \bibliographystyle{IEEEtran}
\bibliography{BibLib1} 
\end{document}